\documentclass[journal]{IEEEtran}
\usepackage{standalone}
\usepackage{import}

\usepackage{amsmath,amssymb,amsthm}
\usepackage{qcircuit}
\usepackage{braket}
\usepackage{empheq}
\usepackage{subfig}
\usepackage{xcolor}
\usepackage[numbers]{natbib}

\usepackage{graphicx}
\graphicspath{{figures/}}
\usepackage{dcolumn}
\usepackage{bm}

\usepackage{hyperref}
\hypersetup{
    colorlinks=true,       
    linkcolor=blue,          
    citecolor=red,        
    filecolor=magenta,      
    urlcolor=blue,           
    runcolor=cyan
}

\usepackage{soul}
\setstcolor{red}

\newcommand{\bes} {\begin{subequations}}
\newcommand{\ees} {\end{subequations}}
\newcommand{\ba}{\begin{eqnarray}}
\newcommand{\ea}{\end{eqnarray}}

\newtheorem{theorem}{Theorem}
\newtheorem{theorem*}{Theorem}
\newtheorem{cor}{Corollary}[theorem]
\newtheorem{lemma}[theorem]{Lemma}
\newtheorem{definition}{Definition}

\renewcommand{\t}{\lfloor \frac d 2 \rfloor}

\def\BibTeX{{\rm B\kern-.05em{\sc i\kern-.025em b}\kern-.08em
    T\kern-.1667em\lower.7ex\hbox{E}\kern-.125emX}}
\begin{document}
\title{Compressing Syndrome Measurement Sequences}
\author{\IEEEauthorblockN{Benjamin Anker\IEEEauthorrefmark{1}, Milad Marvian\IEEEauthorrefmark{1}}\\
    \IEEEauthorblockA{\IEEEauthorrefmark{1}Department of Electrical \& Computer Engineering and Center for Quantum Information and Control, University of New Mexico, Albuquerque, NM 87131, USA}
\thanks{This work is supported by NSF CAREER award No. CCF-2237356.}}

\maketitle
\begin{abstract}
    In this work, we analyze a framework for constructing fault-tolerant measurement schedules of varying lengths by combining stabilizer generators, and prove results about the distance of such schedules by combining according to classical codes. Using this framework, we produce explicit measurement schedules sufficient for fault-tolerant error correction of quantum codes of distance $d$ with $r$ independent stabilizer generators using only $O(d \log{r})$ measurements if the code is LDPC, and $O(d \log d \log r)$ measurements if the code is produced via concatenating a smaller code with itself $O(\log d)$ times. In both of these cases the number of measurements can be asymptotically fewer than the number of stabilizer generators which define the code. Although optimizing our construction to use the fewest measurements produces high-weight stabilizers, we also show that we can reduce the number of measurements used for specific examples while maintaining low-weight stabilizer measurements. We numerically examine the performance of our construction on the surface code under several noise models and demonstrate the exponential error suppression with increasing distance which is characteristic of weak fault tolerance. 
\end{abstract}
    
    \emph{This work has been submitted to the IEEE for possible publication. Copyright may be transferred without notice, after which this version may no longer be accessible.}
    \section{Introduction}
    Although quantum computers promise asymptotic advantage over classical computers for certain computational tasks, they are extremely sensitive to noise. As such, quantum error correction is essential to perform reliable quantum computation. This fact has been recognized since the birth of the field of quantum computation and has produced a plethora of proposals for how to achieve reliable computation~\cite{shor1997FT,yoder2016concat,lidar2013qec,knill1999qec, fowler2012surface}. The best studied class of quantum error correcting codes, however, are stabilizer codes \cite{gottesman1997stabilizer}, which are defined by a set of independent commuting stabilizer generators. These stabilizer generators impose constraints on admissible states, or codestates; the set of all codestates, or the codespace, is the joint $+1$ eigenspace of the stabilizer generators. By measuring the stabilizer generators we can determine which constraints are violated and hence diagnose errors on a codestate.
    
    In physical implementations of error correcting codes, it is often the case that measurement takes much more time than one- or two-qubit gates~\cite{acharya2022suppressingquantumerrorsscaling, monroe2021faulttolerantcontrol, foxen2020continuous}. The delay caused by measurement both impacts the clock speed of the quantum computation, and allows for more errors to occur while the measurement is being performed. 
    
    The fact that measurement is slow and noisy motivates one to ask whether there are ways to perform fault-tolerant error correction using fewer measurements. This has been developed in many directions. For instance single-shot error correction \cite{bombin2015single, campbell2019singleshot,stahl2024single,brown2024single,kubica2022single,quintavalle2021single, gu2024single,hillmann2024single} aims to find a decoder that guarantees fault tolerance even when measuring each stabilizer generator only a constant number of times. Work has also been done to choose \emph{which} stabilizers are measured as a function of the previous stabilizer measurement results \cite{berthusen2025adaptive, berthusen2024partial, tansuwannont2023adaptive}. Directly optimizing specific choices of codes is also a viable path to reduce the resources needed \cite{delfosse2020short}. Of course, extracting discrete valued syndromes is a simplifying reduction from more general measurement schemes \cite{mohseninia2020continuous,kumar2019weak}. 
    
    In our work we focus almost exclusively on the number of measurements made. Since an $[[n, k, d = 2t + 1]]$ stabilizer code is defined by $r = n - k$ independent stabilizer generators, one might immediately expect that it is impossible to reduce the number of measurements required to correct $t$ errors below $r$. However, counting distinct errors and quantifying the information needed does not immediately rule this out. Suppose our goal is to distinguish any two errors with weight at most $t$ by observing which constraints they violate. 
    Since the number of errors with at most $t$ is only $\sum_{i = 0}^{t} \binom{n}{i} 3^i  = O(2^{d \log n})$, a na\"ive counting argument implies that we can distinguish these errors with $O(d \log n)$ carefully chosen measurements; this observation is similar to Gottesman's~\cite{gottesman2022review} more general estimate of the number of measurements sufficient to distinguish every space-time error, and related to lower bounds on the number of measurements necessary~\cite{nemec2023hamminglikebounddegeneratestabilizer,aly2007notequantumhammingbound,dallas2022hamming}.
    Although it is not immediately clear that the measurements necessary will be stabilizers, this proposition is  shown to be true by Delfosse et al.~\cite{delfosse2022beyond} for codes with block-length $n$ which scales polynomially with $d$, i.e. codes with parameters $[[O(d^\alpha), \star, d]]$. In this case the number of measurements required scales as $O(d \log d)$.

    The existence of such a sequence of measurements is even less obvious when we ask that the resulting series of measurements is fault tolerant \cite{shor1997FT}. With this demand, it is not enough simply to find a set of stabilizers which distinguishes every data error as our sketched estimate above showed should be possible; rather, we need to find a set of measurements which makes this distinction even in the case of measurement errors. For our purposes we focus on a minimal version of fault tolerance, namely weak fault tolerance~\cite{tansuwannont2023adaptive,delfosse2022beyond}. Roughly speaking, this ensures that low-weight errors propagate to low-weight errors, but contrary to strong fault tolerance makes no guarantee that general, high-weight, errors will produce states close to the codespace. We motivate the utility of weak fault tolerance more in Section~\ref{sec:weak_FT}.  
     Delfosse et al.~\cite{delfosse2022beyond} also prove the surprising existence of \textit{sub-single-shot} fault-tolerant quantum error correction.   
    As originally presented, the result obtained by Delfosse et al.~\cite{delfosse2022beyond} is not constructive, and it is not clear how to efficiently produce a measurement schedule which achieves this scaling.

    In this work, we provide a general framework to produce fault-tolerant measurement schedules for any $[[n, k, d]]$ code which has the property that errors of weight at most $d - 1$ have syndromes with weight at most $cd$, for a given $c$. This framework produces a measurement schedule using only $O(cd \log n)$ measurements, which reduces to the bound given by Delfosse et al. when $n = O(d^\alpha)$ and $c$ is a constant (e.g. when the code is LDPC). Our construction also applies to concatenated codes constructed by concatenating a smaller code with itself $O(\log d)$ times, which produces a bound on $c$ of $\log d$ and a bound on the total number of measurements of $O(d \log d \log n)$.
    
    Our construction is based upon the insight that bounded-weight syndromes can be interpreted as a codeword of a classical code with a bounded number of errors. By checking the parity of subsets of this codeword one can deduce the location of errors, which are just syndrome bits equal to $1$, and hence recover the syndrome of the original error on the data qubits. Checking the parity of a subset of bits is equivalent to measuring a set of stabilizers obtained by combining stabilizer generators. The fact that the syndrome of the quantum code is of bounded weight allows us to choose a classical code with a distance comparable to the quantum code and hence with relatively few parity checks, which correspond to few measurements.

    Our method of combining two codes is somewhat similar to the proof given by Delfosse et al.~\cite{delfosse2022beyond} that sub-single-shot error correction exists, in that the authors considered multiplying the parity check matrix of the code of interest by a random matrix of the correct dimensions. However, choosing a parity check matrix allows us to use the structure of the code to make precise statements about the distance of the resulting sequence of measurements.
    
    The framework we develop is general, and allows for a systematic choice of which stabilizers to measure based upon the trade-offs one is willing to make between the number of measurements and the weight of the stabilizers measured. To achieve the bound given by Delfosse et al., we choose a classical code with $O(d \log n)$ parity checks, which produces a set of stabilizers each with weight $\Theta(n)$. Other choices of classical codes, such as those examined in Section~\ref{sec:alternative}, can produce measurement schedules with more measurements but a better upper bound on the weight of the stabilizers. 
    
    This procedure is related to data-syndrome codes~\cite{ashikmhin2014robus, brown2024datasyndromebch, fujiwara2014datasyndrome}. In the data-syndrome code framework, stabilizer generators of a quantum code are also combined using a classical code, the eponymous data-syndrome code. However, in this application the focus is on accounting for measurement errors; rather than understanding the syndrome of the quantum code as an encoded state, it is understood as the codeword prior to encoding. The simplest example of this is when the codeword, or syndrome, is encoded using the repetition code by repeated rounds of syndrome extraction. For this reason the matrix they use to combine stabilizer measurements is the \emph{generator} matrix of the classical code, rather than the parity check matrix, and is related to other method for making syndrome extraction robust~\cite{premakumar20192designs}. Making syndrome extraction more reliable is also useful, and is in some sense our goal, but this framework does not allow us to make the same kind of statements about the distance and how the number of measurements change.  We expand on the connections and differences of the two methods in Section~\ref{sec:data_syn}. 
    
    Our framework is also related to the idea of metachecks~\cite{campbell2019singleshot,hillmann2024single}. Metachecks encode restrictions on the form a measurement-error-free syndrome can take, and are non-physical in the sense that they are computed in post-processing as opposed to physically measured for violation. That is, every code affected by a set of errors $E = \{e : |e| < d\}$ produces a set of syndromes $S = \{s : \exists_{e\in E} \sigma(e) = s\}$, where $\sigma(e)$ is the syndrome of $e$. Since the stabilizer generators measured to produce the checks are linear (over $GF(4)$) these syndromes can be interpreted as codewords of a classical code defined by a parity check matrix $P$ such that $Ps = 0$ for all $s \in S$. Then measurement errors $u$ produce metasyndromes since $P(s + u) = Pu$. However, these metachecks are again focused on the identification of measurement errors whereas we consider every syndrome $s$ as a version of the $\vec 0$ codeword with some errors, rather than its own distinct codeword. This allows us to find the location of ones in $s$, meaning we deduce the syndrome, instead of the location of $1$s in $u$, which correspond to measurement errors.
    
    The organization of this paper is as follows. In section~\ref{sec:background} we provide background on stabilizer codes and quantum error correction in general, as well as defining weak fault tolerance and motivating why we construct weakly fault tolerant protocols. In section~\ref{sec:main_result} we present the intuition behind our framework, as well as proving the fault-tolerance of our construction. 
    In section~\ref{sec:decoding} we outline the difficulties imposed by decoding our construction and propose solutions to some of them.
    In section~\ref{sec:surface_code} we provide an example of our construction applied to the surface code, and provide numerical results demonstrating exponential error suppression. In section~\ref{sec:alternative} we provide methods to refine the application of our construction. In section~\ref{sec:single_shot} we apply a result by Campbell~\cite{campbell2019singleshot} to show that the fact that our construction produces a set of measurements sufficient for error-correction in one round implies a construction for fault-tolerant error correction with the same number of measurements. In Section~\ref{sec:data_syn}, we clarify the relation of our work to data-syndrome codes and combine the two approaches to produce short, robust, measurement schedules. Finally, in section~\ref{sec:conclusion} we discuss the implications of our work and possible future directions.
    \section{Preliminaries}\label{sec:background}
    In this work, we focus on reducing the number of measurements required for fault-tolerant error correction using stabilizer codes; one should note that our results also directly apply to stabilizer subsystem codes. Here we provide a brief review of error correction, particularly stabilizer codes and their classical analogues linear block codes, as well as fault tolerance.
    
    \subsection{Error Correction}
    To protect information against physical errors, a common strategy is to encode the logical degrees of freedom redundantly into many physical degrees of freedom, for example, a smaller number of logical qubits into a larger number of physical qubits. On these physical qubits, we put constraints -- these constraints reduce the number of physical degrees of freedom to the number of logical degrees of freedom, but provide the ability to diagnose errors. Elements of a stabilizer group are one convenient choice of such constraints.
    
    A stabilizer group is an abelian subgroup of the Pauli group $P^n$ defined by a set of independent stabilizer generators $G = \{g_i\}$ that does not include $-I$. The code space of a stabilizer code is the joint $+1$ eigenspace of the stabilizer generators. Because the stabilizers commute, the codespace is well defined, and because $-I \not \in \langle G \rangle$ it is non-trivial. By measuring the stabilizer generators, we project to either the $+1$ or $-1$ eigenspace of each stabilizer generator; the $-1$ measurements results then tell us which stabilizer generators are violated, and hence which errors have occurred. Because of the projective stabilizer measurements, we do not need to explicitly consider the full spectrum of errors which can physically occur; rather, it is enough to correct Pauli $X$, $Y$ and $Z$ errors. This is because $I, X, Y, Z$ form a basis for all single qubit errors. If we can find a set of stabilizer generators such that each generator is either purely $X$-type or purely $Z$-type, i.e. Calderbank–Shor–Steane (CSS) quantum codes~\cite{calderbank1996css, steane1996css}, we can further reduce the problem to correcting $X$ errors with the $Z$-type checks, and vice versa. 
    
    Given a set of $r$ independent stabilizer generators on $n$ qubits, the number of logical qubits is given by $k := n - r$. One can see this by observing that each independent stabilizer generator divides the dimension of the codespace by $2$. This in turn is because each stabilizer generator has an equal number of $+1$ and $-1$ independent eigenstates. 
Since the stabilizer generators are independent, this implies that $k = n - r$.
    
    Finally, the (non-trivial) logical operators of the code are defined as members of $N(S)\setminus S$, where $N(S)$ is the normalizer of the stabilizer group $S$. That is, logical operators commute with all stabilizer generators, but are not themselves contained in the stabilizer group. The logical operators define the minimum \emph{distance} of the code -- the weight of the lowest weight logical operator, $d$, is the minimum distance of the code. Normally, we refer to the minimum distance just as the distance. Any error of weight less than $d$ can be detected by the code, and any error of weight less than $d/2$ can be corrected by the code; these two statements are equivalent.
    
    Linear block codes form a similar picture classically. As in the quantum case, we encode a relatively small number of logical bits into a relatively large number of physical bits. The constraints we put on the physical bits can be considered as stabilizers composed of a tensor product of only $Z$ operators, but it is more natural to consider them as \emph{parity checks}. A parity check, as the name suggests, checks the parity of some set of bits of the encoded message. The constraint imposed is that the parity of this set of bits should be zero. Independent parity checks are normally considered as rows of a parity check matrix, usually denoted by $H$. Then, if encoded messages are viewed as column vectors, the constraints can compactly described by saying that codewords must be in the nullspace of $H$. If a codeword $m$ suffers an error $e$, then the result of $H(m + e) = He$ is referred to as the syndrome of $e$, and tells us which constraints $e$ violates, in analogy to how measuring each stabilizer generator tells us about quantum errors.
    
    In what follows, we will show that, although a stabilizer code is \emph{defined} by $r$ stabilizer generators, it is possible to perform fault-tolerant error correction using only $O(d \log r)$ measurements, which in general can be much smaller than $r$.
    \subsection{Weak Fault Tolerance}\label{sec:weak_FT}
    In section~\ref{sec:main_result} we will present our construction for a weakly fault-tolerant measurement schedule consisting of $O(d^2 \log d)$ measurements. 
    First, we define weak fault tolerance and motivate our focus on it.
    \begin{definition}
        An error correction procedure is weakly fault tolerant to distance $d = 2t + 1$ if when given an input state which has suffered $r$ errors and suffering $s$ errors during error correction the corrected output state has at most $s$ errors for all $r, s$ such that $r + s \leq t$.
    \end{definition}
    This definition of fault-tolerance is the one we will use in the remainder of this work, except for section~\ref{sec:single_shot}.
    Although weak fault tolerance is not suitable for concatenation, which is commonly used to show that a threshold exists, a weakly fault tolerant error correction procedure still exponentially suppresses logical error rates. To be precise, if noise is local, occurring with probability $p$ on each data qubit before each stabilizer measurement and with probability $p$ on each syndrome qubit before measurement then the logical error probability scales as $p^d$, where $d$ is the distance that the error correction procedure is fault tolerant to \cite{delfosse2022beyond}. This noise model corresponds to a physical implementation of stabilizer measurements which is either transversal or uses flag qubits. Note that this does not imply that a distance $d + 2$ code is necessarily better than a distance $d$ code for any given error rate, as would be the case if $p$ were to be less than the threshold error rate for this family of codes; converting a family of codes with increasing distance into a family of weakly fault tolerant codes with increasing distance does not necessarily produce a family of codes which admits a (positive) threshold. In section~\ref{sec:surface_code} we observe behavior that suggests a threshold exists under a minimum-weight error (MWE) decoder for phenomenological noise; however, even using a less optimal but still weakly fault tolerant decoder, we observe the characteristic exponential error suppression of weak fault tolerance in Figure~\ref{fig:two_step}.
    
    We now make an observation about the number of measurements required to ensure weak fault tolerance. It is relatively easy to see that $O(d^2)$ rounds of syndrome extraction is sufficient to ensure the (strong) fault tolerance of an $[[n, k, d]]$ code \cite{shor1997FT}, but fewer rounds is sufficient for weak fault tolerance. We first prove a lemma about the detectability of errors, using the following definitions.
    \begin{definition}
        We define the accumulated data error up to round $i$, denoted by $\bar e^i$, as $\bigoplus_{e \in \bigcup_{j < i}E_j}e $ where $E_j$ is the set of all data errors that occur during round $j$ of measurement ($E_0$ being errors that occur before error correction begins).
    \end{definition}
    \begin{definition}
        We define an error $e$, consisting of errors both on data and measure (qu)bits at any point in time, as undetectable relative to some sequence of parity check (stabilizer) measurements $\{g_i\}_{i \in I}$ if, when suffering only error $e$, each $g_i$ is measured as $0$.
    \end{definition}
    \begin{lemma}\label{lem:undetectable}
        If $H$ is an $r$ by $n$ parity check matrix, then any error which is undetectable to the sequence of measurements defined by measuring the rows $h_1$ through $h_r$ in order and repeating this sequence of measurements $\gamma$ times either has weight at least $\gamma$, or satisfies $H \bar e^i = 0$ for some $i$.
    \end{lemma}
    \begin{proof}
        Suppose the error is undetectable and $H\bar e^i \not = 0$ for all $i$, i.e. the accumulated data error never becomes a logical operator. We will show that at least one error must occur during each of the $\gamma$ repetitions in this case and hence the total weight of the error must be at least $\gamma$. Suppose that there is one repetition which does not contain an error. Before measuring $h_1$ the accumulated data error satisfies $H\bar e^i \not = 0$, so at least one of the measurements anticommutes with $\bar e^i$. We assume that no errors occur during this repetition, so after measuring each parity check we produce a non-zero syndrome. This contradicts the assumption that $e$ is undetectable to this sequence of measurements, so each repetition must suffer at least one error, meaning there must be at least $\gamma$ errors if $H\bar e^i \not = 0$ for all $i$.
    \end{proof}
    Saying that $H\bar e^i \not = 0$ is equivalent to saying that the accumulated data error is not a logical operator of the code defined by $H$ at any time. This is motivated by the fact that a trivial logical operator does not need to be detected, and a non-trivial logical operator is necessarily of high weight. In particular, repeating the full set of parity checks for a distance $d$ code $d$ times is enough for any error of weight at most $d - 1$ to be detectable, since logical operators have weight at least $d$ by definition.
    
    The fact that any error of weight less than $d$ is detectable shows that this sequence of measurements forms an error correction protocol that is weakly fault tolerant to distance $d$ when combined with a Minimum-Weight-Error (MWE) decoder by appealing to the circuit distance, which satisfies the following lemma. 
    \begin{lemma}[Circuit Distance, Delfosse et al. \cite{delfosse2022beyond}]
        For any circuit error $\varepsilon$ such that $|\varepsilon| \leq (d_\text{circ} - 1)/2$ the MWE decoder  corrects the input error.
        \end{lemma}
    The circuit distance $d_\text{circ}$ is analogous to the distance of a quantum code in that it is defined as the weight of the minimum weight undetectable error in the circuit. Therefore it is intuitive that any error of weight less than half the circuit distance is correctable by a MWE decoder, as the lemma states.

    \section{Compressed Syndromes}\label{sec:main_result}
    We now introduce the intuition for our method of combining the performance of two codes through compressing the syndromes of the first code with the parity check matrix of the second.
    Consider the parity check matrix $H_{r\times n}$ for some code of interest. Our goal is to find some sequence of measurements with length $r'<r$ consisting of linear combinations of rows of $H$.
    That is, we wish to find a matrix $P_{r' \times r}$ such that the measurement matrix $H_m = PH$ has fewer rows than $H$, but has the same asymptotic distance.
    
    We observe that we can relate the distance of $H_m$ to the distances of $P$ and $H$. Instead of directly considering the measurements defined by $H_m$, we consider $P$ acting upon the syndrome $s = He$ for each error $e$. If $P$ is the parity check matrix for another code and $s$ is sufficiently close to some fixed codeword of the code defined by $P$, then it compresses this syndrome, allowing us to identify the location of each bit set to $1$. Identifying the non-zero bits  of the syndrome of course just tells us which stabilizer generators anticommutes with the error. We now formalize this observation.
    \begin{definition}
        We define the maximum syndrome weight of a code as
        \ba
        w= \max_{|e|< d} |H e|,
        \ea
        where $H$ is the parity check matrix of the code and $d$ is its distance, and $|\cdot|$ denotes the Hamming weight.
    \end{definition}

    \begin{theorem}\label{lem:detectability}
        Let $H_{r, n}$ be the parity check matrix for an $[n, k, d]$ code with maximum syndrome weight $w$, and let $P_{r', r}$ be the parity check matrix for an $[r, \star, w + 1]$ code. Then if $PHe = 0$ either $|e| \geq d$ or $He = 0$.
    \end{theorem}
    \begin{proof}
        Suppose that $He \not = 0$ and we wish to show  $|e| < d$ for the sake of contradiction. We have $0 < |He|\leq w + 1$ by assumption. Denoting $He$ by $s$, we have $Ps = 0$ for $0 < |s| < w + 1$. But this contradicts the assumption that $P$ defines a code with distance $w + 1$. So $|e| \geq d$.
    \end{proof}
    The new total number of measurements required is equal to $r'$ which a priori has no relation to $r$, meaning that we can reduce the number of measurements needed for detecting errors.
    This lemma can easily be extended to CSS codes by considering $H_x$ and $H_z$ separately, or to general quantum stabilizer or subsystem codes by defining the stabilizers by a matrix over $GF(4)$.

    In section~\ref{sec:surface_code} we will use the parity check matrix $P$ from the classical Bose-Chaudhari-Hocquenghem (BCH)~\cite{bose1960bch, hocquenghem1959bch} code to compress the measurement sequence. We describe the BCH code in more detail in appendix~\ref{appendix:BCH}.
    We choose this code because it allows us to achieve the asymptotic scaling of $r'=O(w \log (r))$ in the number of measurements, leading to an asymptotic reduction in the required number of measurements.

    We now consider the application of this scheme to two important families of codes: LDPC codes and concatenated codes.

    By definition, each bit (or qubit) in an LDPC code is involved in only a constant number of checks, which implies that the maximum syndrome weight remains asymptotically small.

    \begin{lemma}
        For any  distance $d$ code in which each (qu)bit participates in at most $c$ many checks, 
        the maximum syndrome weight $w$ satisfies  $w \leq c(d - 1)$.
    \end{lemma}
    \begin{proof}
        Any weight one error is in the support of at most a constant number of checks $c$, and hence produces syndrome of weight at most $c$. By the triangle equality this implies that any error $e$ produces a syndrome of weight at most $c|e|$. By assumption $c|e| \leq c (d - 1)$.
    \end{proof}

     In Section \ref{sec:surface_code}, we analyze in detail the performance of the code obtained by applying our scheme to the surface code, a notable example of a qLDPC code. LDPC codes are not the only family of codes which satisfy the sufficient condition that low weight errors have low weight syndromes; concatenated codes also enjoy such a property.
    \begin{lemma}\label{lem:concat_wt}
        For a code concatenated with itself $m$ times where each qubit in the base code is in the support of at most $c$ checks, any error $e$ has a syndrome of weight at most $|e|cm$.
    \end{lemma}
    \begin{proof}
        We first reduce to the case of single qubit errors giving syndromes of weight at most $cm$ by the triangle inequality.
        We proceed by induction on $m$. For $m = 1$, this condition clearly holds. Now, assuming the assertion holds for $m = x$ we consider the case of $m = x + 1$. Consider the set of stabilizer generators at level $x + 1$ obtained by measuring logical operators at level $x$. It is enough to show that at most $c$ of these anticommute with the error in question. Since we only consider a single qubit error, it must be in the support of only one logical qubit at the top level. By assumption, each logical qubit is only in the support of at most $c$ stabilizers. Therefore at most $c$ of the stabilizers at the $(x + 1)$th level anticommute with the error in question, and the total syndrome weight is at most $|e|cm$.
    \end{proof}
    \begin{cor}
        An $[[N, k, D]]$ code produced by concatenating a $[[n, k, d]]$ with itself $\log_d D$ layers, has maximum syndrome weight of $w \leq c D \log_d{D}$.
    \end{cor}
    That is, using Theorem~\ref{lem:detectability} we can pick $P$ with a distance only $c \log_d{D}$ times the desired distance to produce a measurement scheme for concatenated codes. Since the number of independent stabilizers for a concatenated code scales as $n^\ell$ (concatenating an $[[n, k, d]]$ code with itself $\ell$ times), this leaves room for certain codes to save many measurements. By Lemma~\ref{lem:undetectable}, repeating the measurement schedules produced by this construction $d$ times is sufficient to produce a measurement schedule weakly fault tolerant to distance $d$.
    
   We now identify the asymptotic regimes where our construction reduces the number of required measurements for concatenated codes. (In section~\ref{sec:surface_code} we demonstrate the performance of the scheme on small codes.) Suppose we apply our construction to a code produced by concatenating an $[[n, k, d]]$ code with itself $m$ times. Since our construction requires $O(m d^m \log n^m)$ measurements per round, if we are to reduce the number of measurements below the number of stabilizer generators, we need $m d^m \log n^m  < n^m$, or $m^2 d^m \log n < n^m$. Suppressing constants, we need $\left(\frac n d\right)^m$ to scale faster than $m^2$, or $m \log \frac n d > 2 \log m$. Clearly this holds for all parameters, showing that our scheme gives a measurement number advantage at some level of concatenation. However, the level of concatenation necessary depends heavily on the ratio between $n$ and $d$ in the base code.
    
    As an example, first we consider the $[[5^m, 1, 3^m]]$ code. For this code, the level of concatenation necessary to see a performance advantage is $m = 15$. Clearly this is not a practical regime for concatenation, since it corresponds to a distance of several million. For the $[[7^m, 1, 3^m]]$ concatenated Steane code, the obtained code is less impractical in that after $6$ levels of concatenation our construction yields a measurement advantage. For practical levels of concatenation, perhaps $m = 3, d = 3$, we could concatenate a code with approximately $10$ physical qubits in order to see a measurement advantage.

    The surprisingly large asymptotic reduction in the number of required measurements is achieved at the expense of potentially increased measurement complexity and in particular increase in the weight of the measured stabilizers. 
    Depending on the classical code used to compress the syndrome, the weights of the prescribed stabilizers may be much higher than the original stabilizer generator weights. For example, the $[n - k, \star, \star]$ BCH code generically has parity checks of weight $\Theta(n - k) = \Theta(n)$, meaning that in the worst case the stabilizers produced by our construction have weight $\Theta(n)$. 
    
    This is not an insurmountable challenge. We observe that the freedom to choose the representative of the parity-check matrix can be exploited to reduce the weight of each measured stabilizer. Furthermore, by carefully ordering the stabilizer generators of the quantum code we can ensure that the combination of stabilizer generators chosen cancel well. In Section~\ref{sec:distance_preserving} we also show that for certain codes, the impact of the high weight stabilizers can be mitigated by exploiting the structure of the quantum code.
    
    \section{Decoding}\label{sec:decoding}
    Our proof of a weakly fault-tolerant error correction procedure implicitly assumes the use of a fault-tolerant decoder. In the simulations in Section~\ref{sec:surface_code}, we use an in-house MWE decoder to decode the syndromes produced by our construction. This decoder is fault tolerant, but is also extremely inefficient, in that it is a brute force decoder. Should one wish to implement this construction in practice, it is necessary to develop a more efficient decoder. 
    
    Heuristically one might expect that decoding the outer classical code to obtain the syndrome for the inner code, then using a standard quantum decoder to decode this syndrome is sufficient. However this is not necessarily the case because of the impact of measurement errors. Consider obtaining a syndrome $s'$ from measuring the stabilizers produced by our construction, which in the absence of measurement errors can be decoded to produce $s$, the syndrome of the quantum code. However, in the presence of measurement errors $e_m$, we instead observe $s' + e_m$. Even if the measurement error $e_m$ is low weight, decoding $s' + e_m$ can produce a syndrome $\sigma$ which is far from $s$, i.e. $|\sigma + s| > |e_m|$. This means that this two-step decoding procedure is not necessarily fault-tolerant.
    
    However, by adjusting the form of the stabilizers, a two-step decoder can be made resistant to measurement errors.
    
    \begin{theorem}
    Let $H$ be a parity check matrix of a code with  maximum syndrome weight $w$ and $P_c$ be the canonical form of  the (compressing) parity check with distance $w$, i.e., $P_c = [I | P']$. 
    The two-step decoder described above applied to the measurement matrix $H_c = P_c H$ can correct all errors of up to weight $(w - m)/2$ (which are not mid-circuit errors) together with $m$ measurement errors.
    \end{theorem}
    \begin{proof}
        Suppose we suffer an error $e$, with a component $e_d$ on the data qubits and a component $e_m$ on the measurement qubits. We assume that $|e| \leq d/2$ and that $e_d$ consists only of errors between rounds of syndrome extraction (or before all rounds).

        Now consider measuring the stabilizers defined by $H_c$ and applying a MWE classical decoder to the measured syndrome. The observed syndrome will be of the form $P_c s + e_m$, where $s = H e_d$. 
        Because $P_c$ is in canonical form, we can find a bitstring $e_m'$ such that $P_c e_m' = e_m$ and $|e_m'| \leq |e_m|$.
        We can then rewrite $P_cs + e_m$ as $P_c(s + e_m')$. Since $s + e_m$ is at most weight $w$ by assumption, a minimum weight decoder applied to the observed syndrome can perfectly identify the locations of ones in $s + e_m'$, that is, deduce the syndrome $s$ with at most $|e_m'|$ errors.

        Therefore the $|e_m|$ physical measurement errors can be interpreted as measurement errors obtained from measuring stabilizer generators when passing the decoded syndrome to the quantum decoder. Since we assume the quantum decoder to be resistant to measurement errors of up to weight $d/2$ the entire two-step decoder is fault-tolerant under this noise model.
    \end{proof}
    The reason this decoder fails  when considering mid-circuit errors is that a mid-circuit error is generically only equivalent to at most $r/2$ measurement errors (instead of a constant number of measurement errors), meaning $e_m$ (and hence $e_m'$) may not be low-weight. Equivalently, mid-circuit errors produce syndromes which are far from the all-zeros codeword the classical code expects. It also may be illuminating to write $P_c(s + e_m')$ as $P_c(H e_d + e_m')$. Then we can intuitively see how the term we get from decoding the outer syndrome according to $P_c$ is just the syndrome of the data error, $He_d$, with some measurement errors $e_m'$.
    
    In numerical experiments, we observe that this two-step decoder does correctly handle data errors and measurement errors up to half the distance of the code. These numerical experiments are explored in section~
    \ref{sec:two_step_numerics}.

    \section{Example: surface code}\label{sec:surface_code}
    Arguably the most famous LDPC code is the surface code. A surface code of distance $d$ is defined by $d^2$ stabilizer measurements. Here we use the BCH  code to choose the stabilizers we measure to preserve a surface code state. The distance $2d$ BCH code on $d^2$ bits has $\lceil 2d \log_2(d^2 + 1)\rceil$ parity checks, meaning that by compressing the surface code with the BCH code we only need $O(d \log d)$ measurements rather than $O(d^2)$ measurements needed to even define the surface code. Repeating the resulting sequence of measurements $d$ times we achieve weak fault tolerance to distance $d$ by measuring $O(d^2 \log d)$ surface code stabilizers.
    
    For example, the distance $17$ surface code on $289$ qubits is defined by $288$ independent generators which are usually measured ${\sim}17$ times for a total of $4896$ measurements to produce a positive threshold. Using our construction, it is enough to measure $278\times 17 = 4726$ operators to ensure weak fault-tolerance.
    
    \subsection{Numerical Experiments on the Surface Code}\label{sec:numerics}
    While the sequence of measurements produced by our construction is weakly fault tolerant, this does not mean that there exists a positive threshold under circuit-level noise for arbitrary circuits implementing these measurements, similar to the fact that a threshold is not guaranteed under arbitrary decoding as we saw in section~\ref{sec:decoding}. However, repeated rounds of weakly fault-tolerant syndrome extraction are sufficient to imply exponential noise suppression when considering either phenomenological noise, or a noise model in which errors are uncorrelated, as in section~\ref{sec:weak_FT}. The latter noise model roughly corresponds to circuit-noise when stabilizer measurements are performed transversally or with the help of flag qubits~\cite{reichardt2018flag})~\cite{delfosse2022beyond}. The remainder of this section is devoted to numerical experiments examining the performance of our construction under various noise models.
    
    In our simulations, we measure all stabilizer operators produced by our construction, and attempt to determine whether there has been a logical error based on the syndrome. Deducing which logical error occurred is enough to correct the physical errors which may have occurred, since we can always return to the codespace by finding any error with the same syndrome. After returning to the codespace, we just need to undo any logical operations applied. In practice, both of these operations are usually delayed as long as possible, until a non-Clifford gate is to be applied, only updating the Pauli frame in the classical decoder up until that point \cite{chamberland2018frame, knill2005realistically}.
     Since the surface code is symmetric with respect to $X$ and $Z$ checks, for our simulations we just measure the $X$ checks produced by our construction and attempt to correct logical $\overline Z$ errors. The code used for these simulations uses Stim~\cite{gidney2021stim} and is available on \href{https://lobogit.unm.edu/banker/number-of-measurements}{LoboGit}. Throughout all experiments to estimate the logical error rate at the physical error rate of $p$ we have taken at least $10\times 1/p$ samples. 
    
    First, we consider the simple case of the code-capacity threshold, to measure how one round of this measurement schedule performs \emph{as a code}, ignoring fault-tolerance for the moment. In this noise model, we only introduce noise on the data qubits, and only before any part of the circuit for the single round of syndrome extraction is implemented. In Figure~\ref{fig:surface_code_cap} we observe an apparent threshold at approximately $15\%$ for this scenario.
    \begin{figure}[ht]
    \centering
    \includegraphics[width=0.5\textwidth]{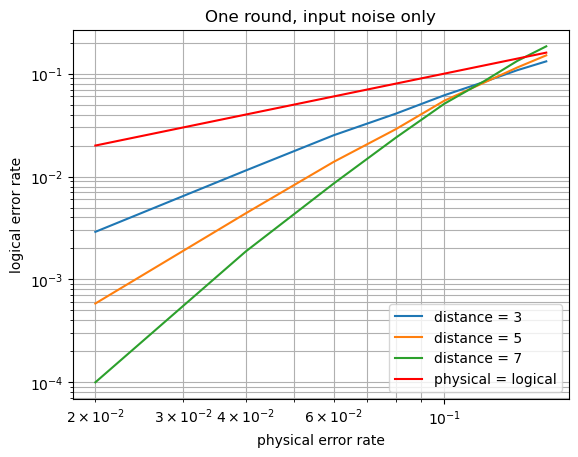}
    \caption{\label{fig:surface_code_cap} Logical vs physical error rates under code-capacity noise for distance $3, 5$ and $7$ using one round of our proposed measurement schedule.}
    \end{figure}
    
    We then move to the phenomenological noise model, i.e. where depolarizing noise on the data qubits with strength $p$ still comes before each round of measurement, but where measurement errors are flipped with probability $p$ as well. In this model, we see a lower threshold, as we should expect to see. Figure~\ref{fig:surface_code_phenom} shows the observed threshold for both one round of syndrome extraction for distances $3, 5, 7$ and $3$ rounds of syndrome extraction for distance $3$; we have omitted higher distances due to computational constraints. Interestingly, a threshold appears to exist even when restricting ourselves to a single round of syndrome extraction. We posit that this is because of the fact that single qubit data errors produce high weight syndromes in our construction, while measurement errors produce weight one syndromes; this distinction allows a MWE decoder to separate the cases relatively easily. 
    \begin{figure}[ht]
    \centering
    \subfloat[][One Round]{
        \includegraphics[width=0.23\textwidth]{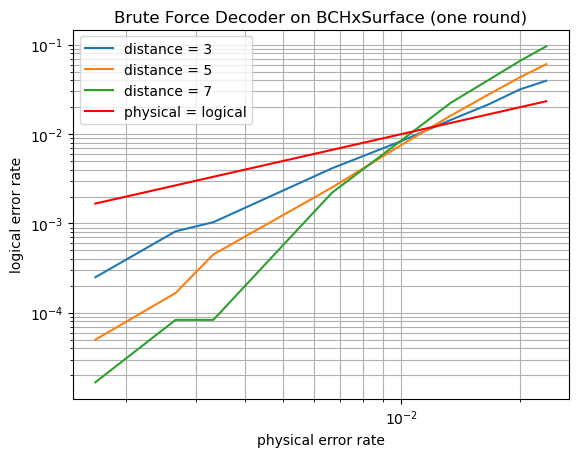}
    }
    \subfloat[][$d$ rounds]{
        \includegraphics[width=0.23\textwidth]{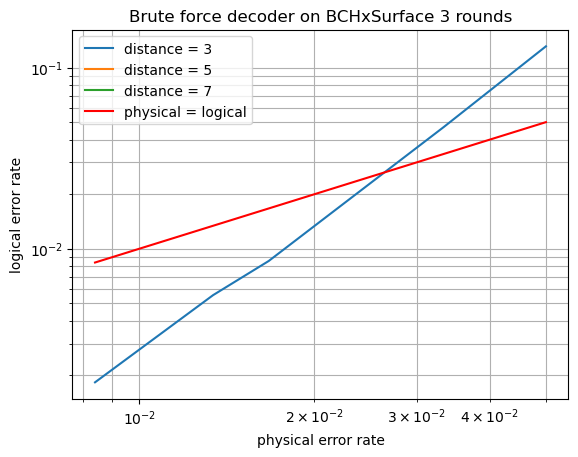}
    }
    \caption{\label{fig:surface_code_phenom} For one round we observe a threshold of approximately $1\%$, while for $d=3$ rounds we observe a pseudo-threshold for distance $3$ at approximately $3\%$. 
    }
    \end{figure}
    
    Finally, we consider the following noise model: before each stabilizer measurement, single qubit depolarizing noise with parameter $p$ is applied to every qubit -- in addition, measurement errors occur with probability $p$. This roughly models the scenario in which syndrome extraction is performed transversally, or with the assistance of flag qubits, so that errors on the syndrome qubit(s) do not propagate to the data. Due to computational constraints imposed by the MWE decoder, we have not simulated the performance of our measurement schedule for distances higher than $3$, but observe a pseudo-threshold in Figure~\ref{fig:surface_code_uncor}. 
    \begin{figure}[ht]
    \centering
    \includegraphics[width=0.5\textwidth]{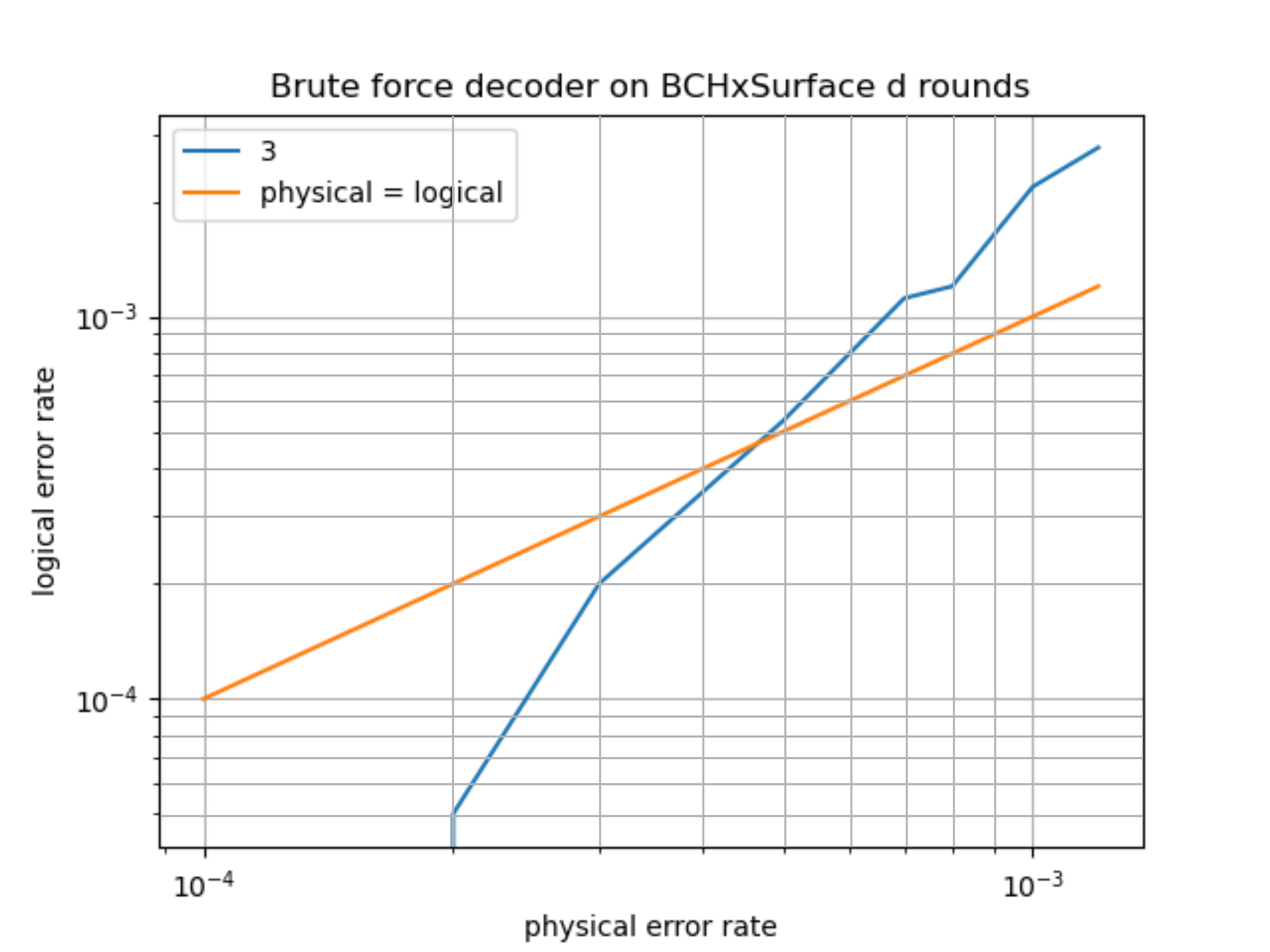}
    \caption{\label{fig:surface_code_uncor} For the noise model consisting of uncorrelated noise on all qubits before each stabilizer measurement and uncorrelated measurement bit flips, we observe a pseudo-threshold at approximately $0.03\%$ for distance $3$.}
    \end{figure}
    \subsection{Numerical Experiments using an Efficient Decoder}\label{sec:two_step_numerics}
    Interestingly, however, we do not observe the same qualitative behavior using the two-step decoder outlined in section~\ref{sec:decoding} as we observe using the brute-force decoder. 
    Instead, there are broad ranges of parameters for which the distance $d$ measurement schedule performs better than the distance $d + 2$ measurement schedule, while both perform better than an unencoded qubit, as shown in Figures~\ref{fig:two_step} and~\ref{fig:two_step_full_noise}. 
    
    \begin{figure}
        \includegraphics[width=\columnwidth]{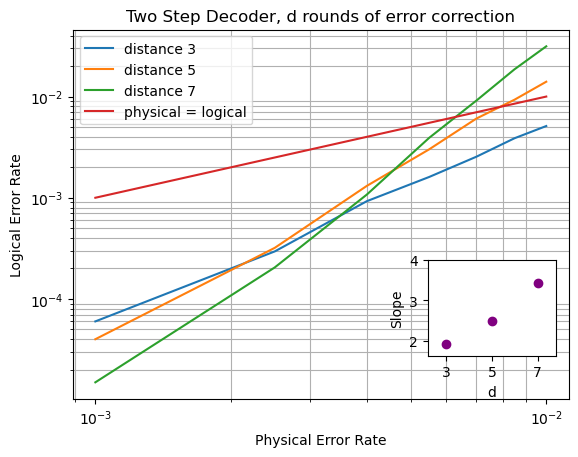}
        \caption{\label{fig:two_step}An example of exponential error suppression without a threshold, where we have used an efficient two-step decoder. This experiment uses a noise model with noise at the beginning of each round and before each measurement, repeated over $d$ rounds of error correction. Slopes inset.}
    \end{figure}
    \begin{figure}
        \includegraphics[width=\columnwidth]{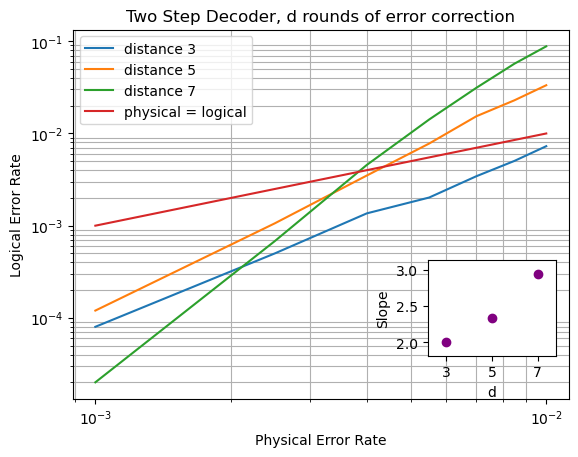}
        \caption{\label{fig:two_step_full_noise} Another example of exponential error suppression without a threshold using an efficient two step decoder. This experiment uses a noise model with one- or two-qubit depolarizing noise after every Clifford, and bitflip noise on every measurement result. Slopes inset.}
    \end{figure}
    
    This is due to the interaction between the classical and quantum decoders; in particular, the fact that the classical decoder performs relatively poorly on the noise model induced by $H$. This in turn is because standard implementations of a classical MWE decoder are blind to the noise model, only finding the minimum weight correction in terms of number of corrected bits. In Figure~\ref{fig:induced}, we have isolated the performance of the classical decoder as it operates on the induced noise model.
    \begin{figure}
        \includegraphics[width=\columnwidth]{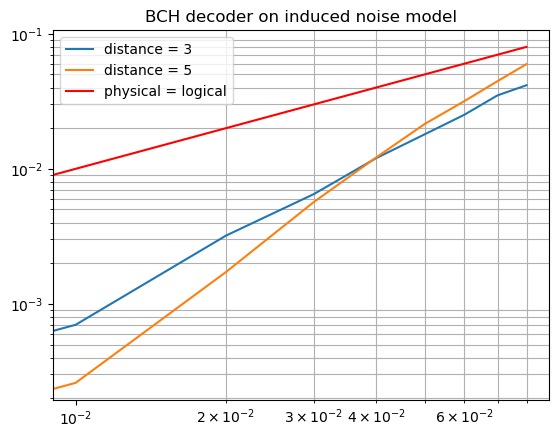}
        \caption{\label{fig:induced} The performance of the classical decoder on the noise model induced by considering syndromes with respect to the surface code with iid errors. The qualitative similarity to Figure~\ref{fig:two_step}, i.e. the large regime in which distance-3 has a lower error rate than distance-5, which has a lower error rate than the distance-1/unencoded state, is due to the fact that the majority of the failures to decode correctly at low error rates are caused by the classical decoder.}
    \end{figure}
    
    For example, suppose the classical decoder correctly produces the original syndrome with respect to stabilizer generators as long as it is at most weight $d/2$ and incorrectly identifies the syndrome otherwise. This is an approximation of the actual behavior of a minimum weight error decoder for the BCH code of distance $d$. Then the classical decoder produces incorrect syndromes for the quantum decoder to operate on, for any error consisting of greater than $d/4$ localized patches of errors, and correct syndromes otherwise. The surface code decoder correctly identifies the Pauli frame when each patch of localized errors is no larger than $d/2$. Therefore, the effect of the two decoders combined is roughly to limit the number of error patches to $d/4$ and the size of each error patch to $d/2$. This is in contrast to taking either decoder independently, which would give a larger number of correctable errors -- in particular, a matching based quantum decoder~\cite{higgott2021pymatching} can correct any number of localized patches of error on the surface code, leading to its naturally high threshold~\cite{dennis2002memory}.
    
    Ideally, one would use an efficient classical decoder which gives MWE corrections with respect to the weight of the error which produces that (quantum) syndrome, not the syndrome itself. In general such a decoder is hard to create, though, and we leave it as an open question whether such a decoder exists in general.
    \subsection{Stabilizer Measurement Form}\label{sec:distance_preserving}
    The simulation results in section~\ref{sec:numerics} suggest that, with sufficient care given to the implementation of the stabilizer measurements, $\Theta(d^2 \log(d))$ stabilizer measurements are sufficient to produce a exponential error suppression for the surface code, compared to the standard $\Theta(d^3)$ measurements of the stabilizer generators. In fact, for the surface code, our construction can be made to inherit the property that Shor style or flagged syndrome measurements are unnecessary with the correct CNOT/CZ ordering. 
    
    For the surface code, one can use the geometric structure of the logical operators to ensure that the only error which propagates from the syndrome qubit to an error with weight greater than one on the data is perpendicular to any logical operator it lies in the support of. This is illustrated in Figure~\ref{fig:surface_stabilizer_measurement}. This means that errors on the syndrome qubit contribute at most one error to a logical error, no matter where they occur. This is actually a more general property of hypergraph product codes~\cite{manes2023distancepreserving}; the rotated surface code we have been considering is closely related to one such hypergraph product code, namely the (unrotated) surface code.
    
    The stabilizers we measure are chosen according to the multiplication of a classical parity check matrix to the stabilizer matrix, producing a stabilizer as the product of stabilizer generators. By neglecting to cancel consecutive CNOTs with the same control and target in each such product when implementing these measurements, as in Figure~\ref{fig:construction_FT_form}, we can ensure that errors on the syndrome qubit have the same perpendicularity property in our construction. Implementing the measurements this way, however, may require more CNOTs than measuring the original stabilizer generators, and the relative importance of ancilla qubits versus two qubit operations must be weighed.
    \begin{figure*}[ht]
        \centering
        \raisebox{-0.5\height}{\includegraphics[width=0.3\textwidth]{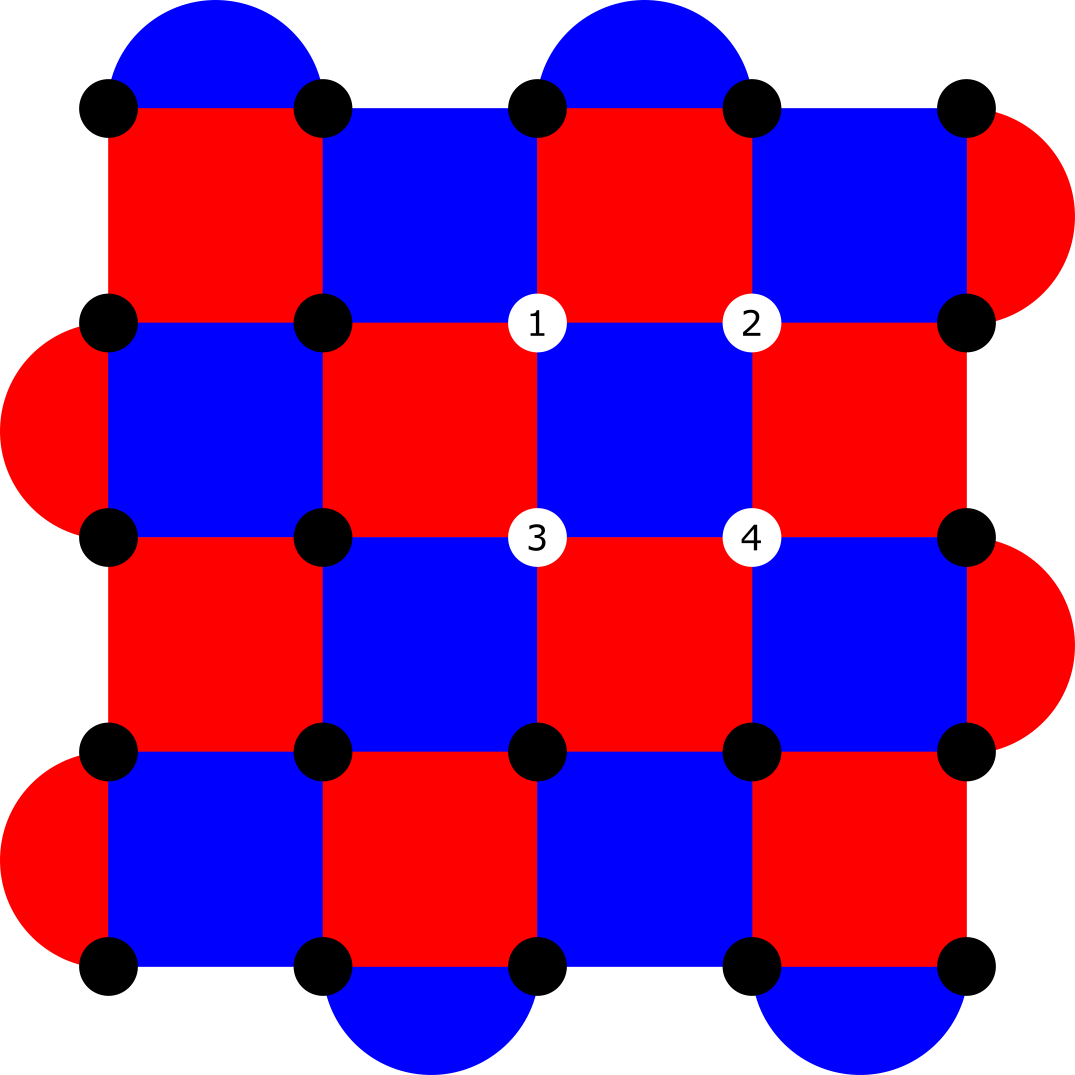}}
        \raisebox{-0.5\height}{\resizebox{0.55\textwidth}{!}{
\hspace{3.2em}
\begin{tabular}{c}
    \vspace{-.6em}\\
    \Qcircuit @C=1em @R=0.3em @!R{
        \lstick{\ket 0} &\gate{\alpha}& \targ    &\gate{\beta}& \targ     &\gate{\gamma}& \targ     &\gate{\delta}& \targ   &\gate{\lambda}& \measureD{Z} \\
        \lstick{A}      &\qw&\ctrl{-1} &\qw& \qw       &\qw& \qw       &\qw& \qw     &\qw& \qw & \\
        \lstick{B}      &\qw&\qw       &\qw& \ctrl{-2} &\qw& \qw       &\qw& \qw     &\qw& \qw & \\
        \lstick{C}      &\qw& \qw      &\qw& \qw       &\qw& \ctrl{-3} &\qw& \qw     &\qw& \qw & \\
        \lstick{D}      &\qw& \qw      &\qw& \qw       &\qw& \qw       &\qw&\ctrl{-4}&\qw& \qw & \\~\\
}
\vspace{.2em}    
\\  
\end{tabular}
}}
    \caption{\label{fig:surface_stabilizer_measurement} Consider a surface code, with $X$ checks shown in red and $Z$ checks shown in blue. The logical $Z$ operator consists of physical $Z$ operators on a vertical line of qubits. If we perform the circuit on the right with qubits A, B, C, D corresponding to qubits on the surface code patch labeled 1, 2, 3, 4, we perform a measurement of the $Z$ stabilizer. No matter which order we measure, a $Z$ error at locations $\alpha, \beta, \delta$ or $\lambda$ will propagate to a weight $0$ or $1$ error up to a stabilizer. However, a $Z$ error at location $\gamma$ propagates to both qubits C and D. If we connect qubits such that the last two qubits measured are $3$ and $4$ (i.e. qubits C and D are qubits $3$ and $4$), however, this only contributes one error to the logical $Z$ operator, and is distance preserving and effectively fault tolerant.} 
    \end{figure*}
    
    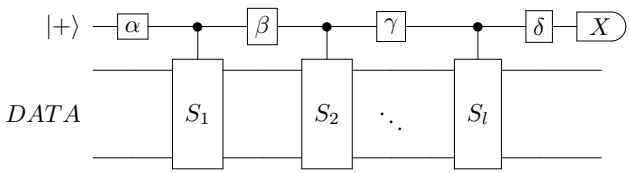
\begin{figure}[ht]
    \centering
            \resizebox{0.5\textwidth}{!}{
            
\hspace{3.2em}
\begin{tabular}{c}
    \vspace{-.6em}\\
    \Qcircuit @C=1em @R=0.3em @!R{
        \lstick{\ket +} &\gate{\alpha}& \ctrl{1}          &\gate{\beta} &  \ctrl{1}         &\gate{\gamma} & \qw        & \ctrl{1}          &\gate{\delta}& \measureD{X} \\
                        &\qw          &\multigate{2}{S_1} &\qw          & \multigate{2}{S_2}&\qw           & \qw        &\multigate{2}{S_l} &\qw           & \qw & \\
        \lstick{DATA}   &             &                   &             &                   &\ddots        &            &                   &              &     & \\
                        &\qw          &\ghost{S_1}        &\qw          & \ghost{S_2}       &\qw           & \qw        &\ghost{S_l}        &\qw           & \qw & 
}
\vspace{.2em}    
\end{tabular}
}
    \caption{\label{fig:construction_FT_form} This circuit measures the product $\prod_{1}^l S_l$. If each $S_i$ is implemented according to the CNOT schedule given in Figure~\ref{fig:surface_stabilizer_measurement}, errors happening during each $S_i$ will contribute at most one error to a logical operator, while errors at locations $\alpha, \beta, \gamma, \delta$ will propagate to stabilizers, making this measurement circuit fault-tolerant.}
    \end{figure}
    \section{Alternative Compression Schemes}\label{sec:alternative}
    Thus far, we have focused on compressing the results for all $\ell = n - k$ stabilizer generators using either one or two classical codes; one for $H_x$ and one for $H_z$ when we consider CSS codes, or one overall when we consider the code over $GF(4)$. From Theorem~\ref{lem:detectability} this requires a code of distance $2d$ to preserve a distance of $d$ for the surface code, since each qubit is in the support of up to two stabilizer generators of the same type. However, we can also apply our compression scheme on subsets of the stabilizer generators. We have already done this by considering $H_x$ and $H_z$ separately for the surface code, but this is not the limit on the subsets we can consider. 
    \subsection{Alternative Surface Code Compression}
    One example for the surface code is given by breaking $H_x$ into two subsets of approximately equal numbers of stabilizer generators, where the first subset $A$ consists of all $X$ stabilizer generators in odd columns of the surface code, and the second subset $B$ consists of all $X$ stabilizer generators in even columns, as illustrated in Figure~\ref{fig:subsets}. We can then apply our construction to $A$ and $B$ separately, producing two sets of stabilizers. Because any two stabilizer generators $a_1, a_2 \in A$ do not share any qubits in their support, an error $e$ produces a syndrome of weight at most $|e|$ with respect to $A$, the same being true for $B$. Because the weight of the produced syndromes is bounded by $|e|$ instead of $2|e|$, we can compress each subset with a code of distance $d$ instead of $2d$.
    
    \begin{figure}
        \centering
        \includegraphics[width=0.8\columnwidth]{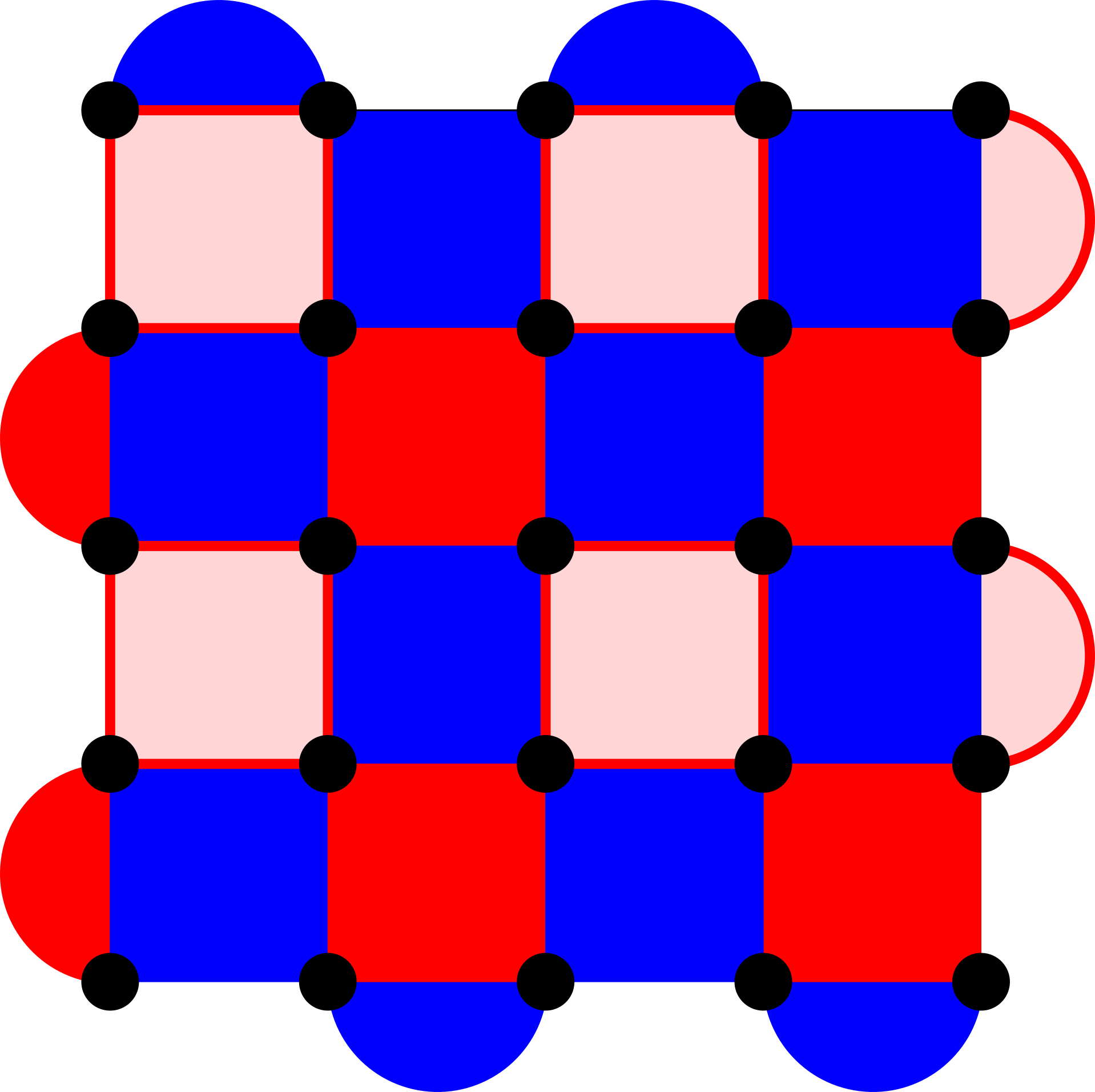}
        \caption{\label{fig:subsets} A partitioning of the $X$-type stabilizer generators of the surface code into disjoint subsets. One subset consists of the pink plaquettes outlined in red, while the complementary set consists of the solid red plaquettes.}
    \end{figure}
    
    Applying the same technique to $H_z$, this scheme then requires $\t \log_2(\frac{d^2 - 1}{4})$ measurements for each of the $4$ subsets compressed separately for a total of $4\t \log_2(\frac{d^2 - 1}{4})$. This is fewer than the total of $4\t \log_2(\frac{n - k}{2})$ measurements given by adding the two sets of $2\t \log_2(\frac{n - k}{2})$ measurements required in the original formulation of our construction for each of the two subsets compressed.
    
    Considering subsets of stabilizer generators with disjoint supports allows for different strategies for compression. For example, we can partition $H_x$ into subsets $\{A_i\}_{i = 0}^{(d - 1)/2}$ where $X$ stabilizers in rows $i$ and $i + \frac{d -1}2$ belong to subset $A_i$. Under this partitioning, an error $e$ can only produce a syndrome of weight at most $|e|$ with respect to each $A_i$. For the distance-$5$ surface code, the partitioning is the same as previously given in Figure~\ref{fig:subsets}.
    
    Then for all $d > 3$ applying the repetition code to each $A_i$ reduces the number of measurements necessary for each subset by $1$. The reduction for each subset is because we can apply a repetition code of distance $d + 1$, which uses $d$ checks to obtain (via decoding the repetition code) the measurement values of $d + 1$ stabilizers. Combining stabilizers according to the repetition code only increases the weight of the stabilizers measured by $4$. Applying the same partitioning to the columns of $Z$ stabilizers yields a total decrease of $d - 1$ measurements for a distance $d$ surface code which would normally require $d^2 - 1$ measurements. While this is not as efficient as the construction based upon the BCH code in terms of the number of measurements required, the fact that the stabilizers are of lower weight makes this construction more practical in many cases. 
    \subsection{LDPC Codes}
    The fact that we can consider subsets of stabilizer generators with disjoint supports is actually a general feature of LDPC codes.
    \begin{lemma}
        Given a set of $\ell = n - k$ LDPC stabilizer generators such that each qubit is in the support of at most $c$ stabilizer generators, we can partition the stabilizer generators into $O(c)$ subsets $\{A_i\}$ such that for any $a_1, a_2 \in A_i$ we have $\text{supp}(a_1) \cap \text{supp}(a_2) = \emptyset$, where $\text{supp}$ denotes the support.
    \end{lemma}
    \begin{proof}
        We proceed in a greedy fashion. Assign an arbitrary stabilizer $a_1$ to subset $A_1$. Repeat the following until no stabilizers satisfy the disjoint support condition: assign the candidate stabilizer $a$ to the subset $A_1$ if $\text{supp}(a) \cap \text{supp}(a_j) = \emptyset$ for all $a_j \in A_i$. Each time we add a stabilizer to subset $A_1$, we remove at most $c$ stabilizers from consideration. Therefore $|A_1| = \Omega(\frac{n - k} c)$. With the remaining stabilizer generators, repeat this entire process with subsets $A_2$ through $A_{O(c)}$ until no stabilizer generators remain unassigned.
    \end{proof}
    
    As long as each of the subsets is large enough, that is as long as $n \geq 2cd$, each of the subsets chosen in this lemma provides sufficient bits for a distance $d$ classical code to exist.
    
    Compressing each subset separately as described reduces the number of measurements required at least from $c d \log(n - k)$ to $c d \log(\frac{n - k}c)$ for a total of at least $cd\log c$ fewer measurements when using the BCH code, although it does not change the asymptotic scaling. It also changes the weight of stabilizers measured from approximately $n - k$ to approximately $\frac{n - k} c$.
    \subsection{Concatenated Codes}\label{sec:concatenated_codes}
    Similarly to LDPC codes, we can find a partition of the stabilizer generators of a $[[N = n^m, 1, D \geq d^m]]$ concatenated code such that each subset has the property that any qubit is in the support of a constant number of stabilizer generators in the given subset.
    \begin{lemma}
        If \[A_i = \{\text{all level $i$ stabilizer generators}\},\] where a level $i$ stabilizer generator acts upon level $i - 1$ qubits, then for any qubit $q$ we have \[\{g_a | \text{supp}{(g_a)} \cap {q} \not = \emptyset, 0 \leq 0 \leq n - k\} \cap A_i \leq c.\]
    \end{lemma}
    To clarify with examples, a level $1$ stabilizer acts upon physical qubits, a level $2$ stabilizer is composed of logical operators of the bottom level code, and so forth.
    \begin{proof}
        Consider an arbitrary qubit $q$. For any level $i$, $q$ is one of the physical qubits making up exactly one of the logical qubits at level $i$. Such a qubit is acted upon by at most $c$ level $i$ stabilizer generators by definition.
    \end{proof}
    Therefore, partitioning into subsets in this manner allows us to compress stabilizer generators from each level separately, using $m$ classical codes each with distance $O(d)$. We now can count the number of measurements required for this scheme. 
    
    It is clear that $|A_i| = \left(\frac{n - 1} n \right)^i (N - 1)$. That is, the first subset includes $\frac{n - 1}n$ of the total stabilizer generators, the next $\frac{n -1} n$ of the remaining stabilizer generators, and similarly for the rest of the subsets. Then the total number of measurements given by applying a BCH code to each subset, measuring stabilizer generators separately when the subsets become small enough so that the BCH code is counterproductive, is at most
    \begin{align*}
        \sum_{i = 1}^{m} cD \log(|A_i|) \leq \\
        m cD \log\left(\frac{n - 1}n (N - 1)\right) < \\
        cD \log D \log N.
    \end{align*}
    At worst, this is the same number of measurements as we proposed in our previous scheme, but could be far fewer measurements for certain values of $n$ and $m$. For instance, if we only apply compression to the stabilizers in the first level and measure the remaining stabilizer generators separately, then the total number of measurements is just $cD \log\left(\frac{n - 1} n N\right) + \frac{1}{n} N$.
    
    \subsection{Tetrahedral Code}
    We can also use this framework to rediscover an already known fact from a new perspective. The tetrahedral code, or the 3D color code of distance three, is a CSS code defined on 15 qubits arranged in a tetrahedral shape \cite{bombin2006color, bombin2015color}. The $X$ stabilizers are weight-8, and the $Z$ stabilizers are weight $4$. The precise shape of the stabilizers is outlined in Figure~\ref{fig:tetrahedral}.
    \begin{figure}
        \centering
        \includegraphics[width=0.9\columnwidth]{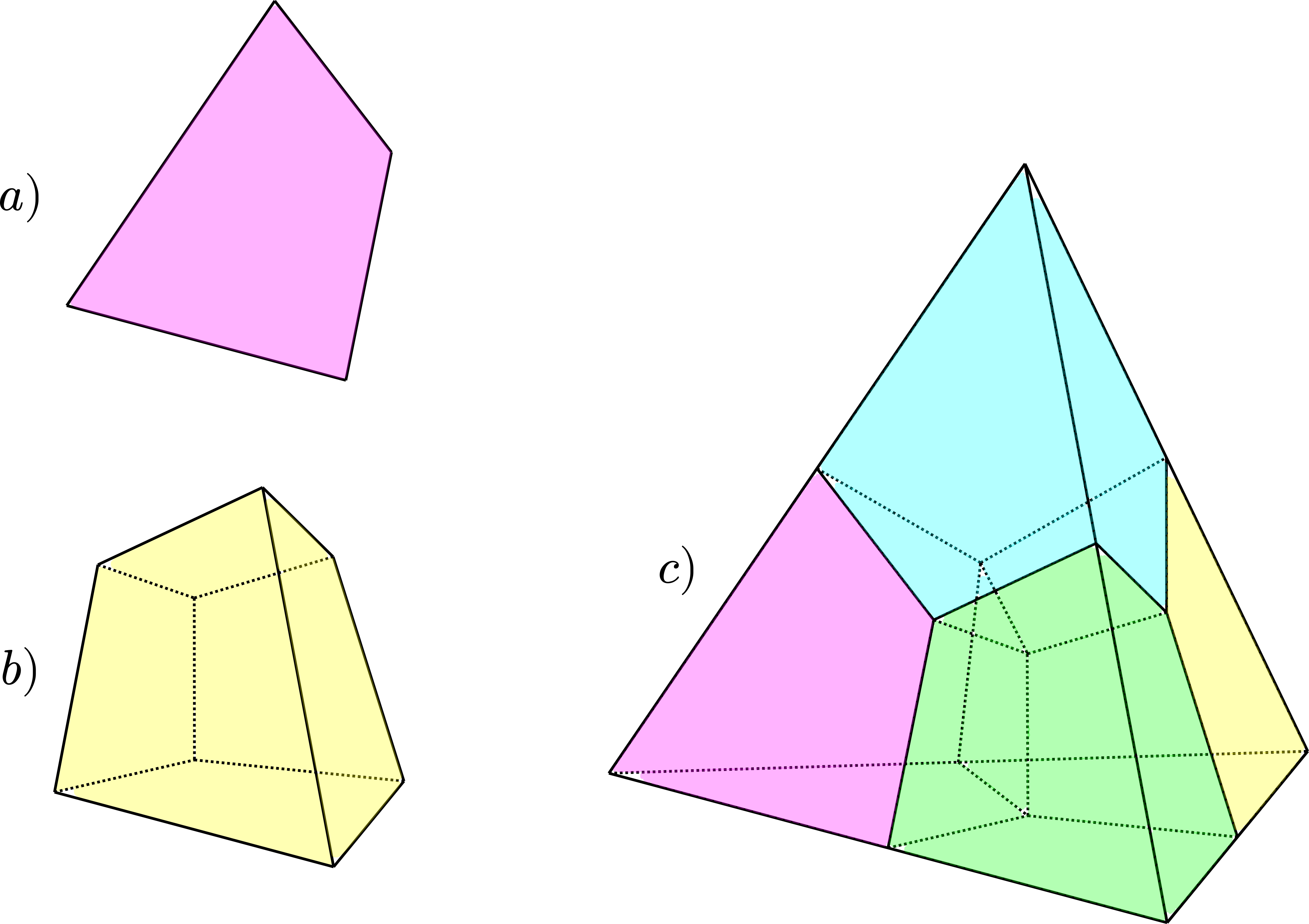}
        \caption{\label{fig:tetrahedral} a) $Z$ stabilizers b) $X$ stabilizers and c) how these elements combine to define the 15-qubit tetrahedral code, with qubits on the vertices.}
    \end{figure}
    Since this code is of distance three, and $Z$ errors commute with $Z$ stabilizers, clearly the measurement of the $X$ stabilizers provides enough information to correct any single $Z$ error. So by symmetry, the $Z$ stabilizers consisting of products of opposite face pairs of $Z$ stabilizer generators, i.e. polytopes shaped the same as $X$ stabilizers and illustrated in Figure~\ref{fig:sufficient_Z}, must also be enough to correct $Z$ errors. This line of reasoning is due to Poulin et al.\cite{poulin2014color}. However, there are only $4$ such stabilizers, as opposed to the $10$ independent $Z$ stabilizer generators.

    \begin{figure}
        \centering
        \includegraphics[width=0.5\columnwidth]{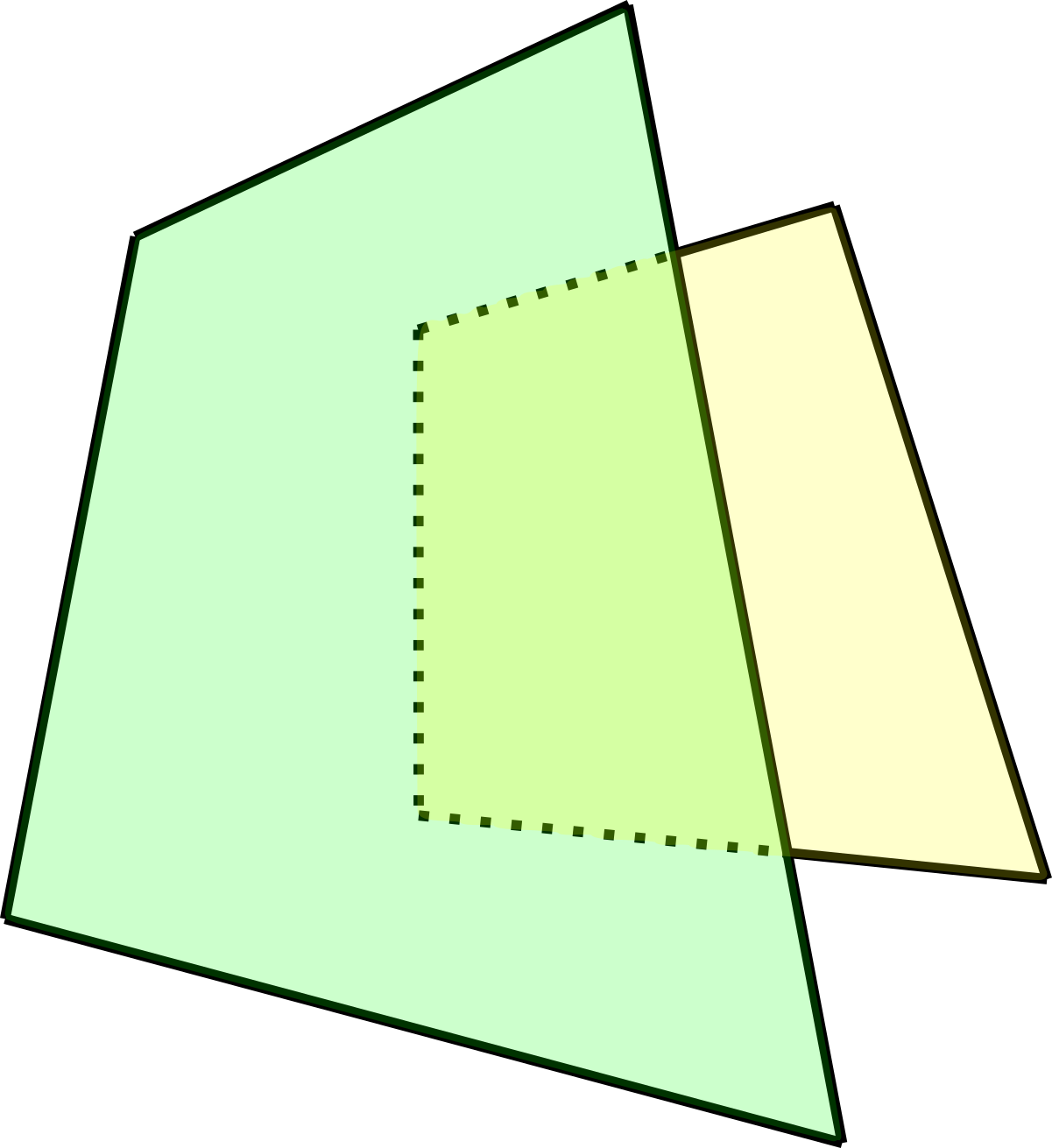}
        \caption{\label{fig:sufficient_Z}An illustration of the fact that the product of two opposite faces of one of the polytopes generates the polytope.}
    \end{figure}

    The fundamental reason for this redundancy may not be immediately apparent. While there are many ways to interpret this fact, one particularly relevant viewpoint is that the 3D color code can be regarded as a gauge fixing of the 3D subsystem color code~\cite{bombin2015color}. This perspective bears a strong connection to our discussion of the measurement schedule defining a subsystem code that the original code is a gauge fixing of, as discussed in section~\ref{sec:single_shot}. This provides a natural explanation for the observed redundancy using our framework.

    To see this, consider taking four sets of three $Z$ stabilizer generators, $S_{GY}, S_{YR}$, $S_{BG}$ and $S_{RB}$, where $S_{GY}$ is the three faces intersected by a line passing perpendicularly through the green and yellow polytopes, $S_{YB}$ the same for yellow and blue, and so on. Three of these subsets are illustrated in Figure~\ref{fig:tetrahedral_subsets}.
    \begin{figure}
        \centering
        \includegraphics[width=0.9\columnwidth]{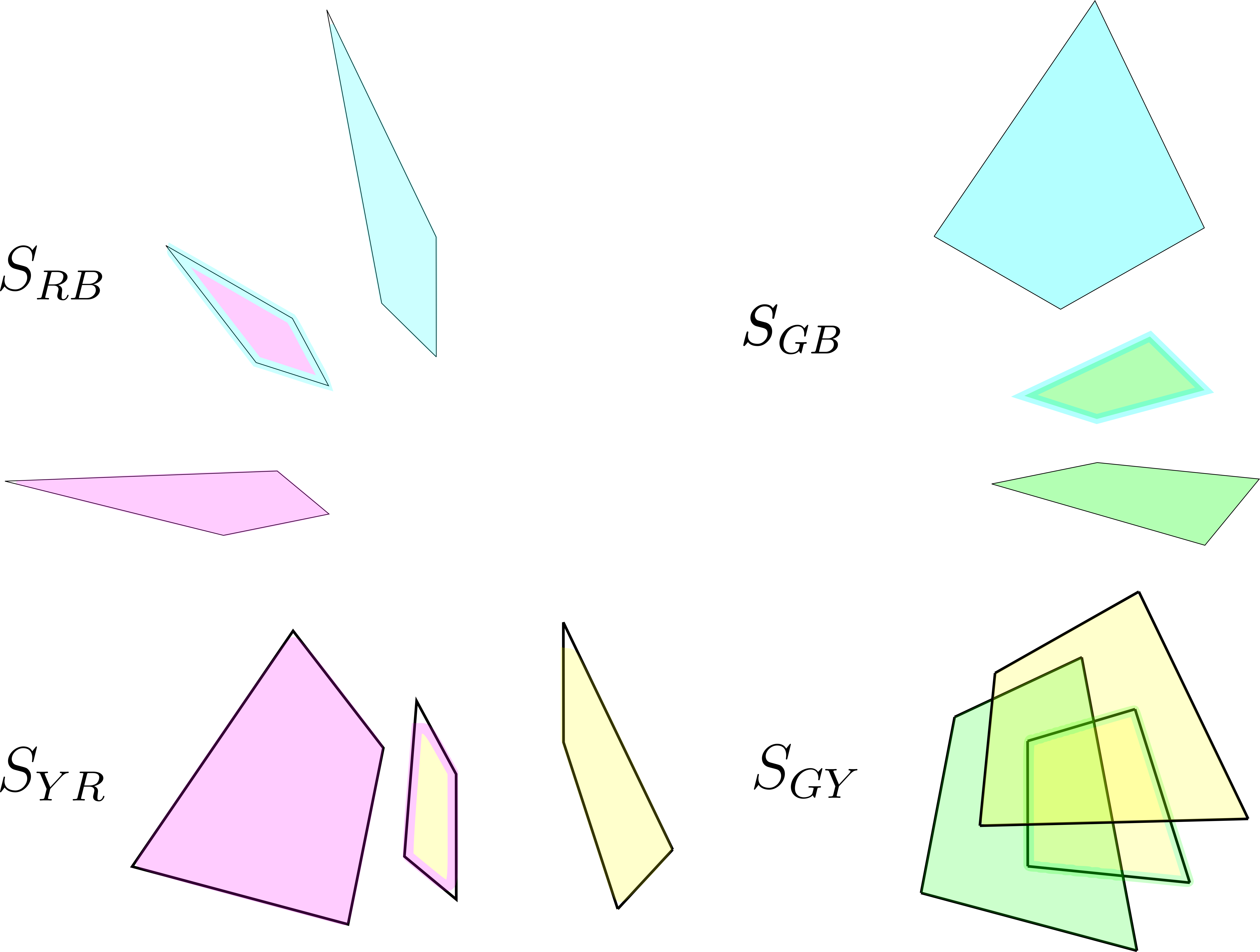}
        \caption{\label{fig:tetrahedral_subsets}Three subsets of the $Z$ stabilizers of the tetrahedral code we consider.}
    \end{figure}

    We now apply the repetition code to each subset. Since there are three stabilizer generators per subset this natively supports a repetition code of distance three. Since none of the stabilizer generators in a subset overlap (within the subset), measuring the checks of the repetition code is enough to deduce the syndrome, with respect to each subset, of any one-qubit data error (and hence correct it). Since the repetition code of distance 3 requires two parity checks, na\"ively we require $2 \frac{\text{checks}}{\text{subset}} \times 4 \text{ subsets} = 8\text{ checks}$. However, we notice that after applying the repetition code checks, we end up with a set of stabilizers which is not independent -- each polytope is measured twice. Measuring them once and reusing the information for each of the subsets is enough to deduce the syndromes of the original weight-4 $Z$ stabilizers. Therefore, we see that our construction applied to the tetrahedral code can find the four $Z$ stabilizers sufficient for the preservation of the distance of the code.
    
    This also suggests another use case for our procedure. We see that, simply by finding and compressing subsets of stabilizers of a given code, we have deduced a gauge code which the tetrahedral code is a gauge fixing of. Since it is known that gauge fixing of gauge color codes can be used for universal quantum computation~\cite{bombin2015color}, this suggests that using the framework presented in this paper to compress other codes may yield candidates for codes related by gauge (un)fixing with desirable properties.

    \section{Single Shot Measurement Schedules} \label{sec:single_shot}
    So far we have focused on compressing the number of measurements required for a single round of syndrome extraction from $n - k$ to $O(d \log(n - k))$, which gives a fault-tolerant measurement schedule of length $O(d^2 \log(n - k))$ as a consequence of the circuit distance~\cite{delfosse2022beyond}. However, due to Earl Campbell~\cite{campbell2019singleshot}, it is known that any set of stabilizers defining a code with distance $d$ can be made single-shot fault tolerant to distance $d$ by some operations on the stabilizers.
    In particular, we use the following fact, which is a direct consequence of Campbell~\cite{campbell2019singleshot} Theorems 1 and 3 and definitions 3 and 5. By $|\cdot|$ we either mean the Hamming weight of a binary vector, or the non-trivial support of a Pauli string, and $\mathrm{wt}_{\min}(\cdot)$ means the minimal support of a Pauli string up to stabilizers.
    \begin{cor}
        For any code with parameters $[[n, k, d = 2t + 1]]$ there exists a check set such that for any binary vector $u$ and Pauli string $E$ satisfying $|u| < t/2$ and $2|u| + |E| < d/2$ there exists a decoder which takes syndrome $s = \sigma(E) + u$ and outputs a recovery operation $E_{\mathrm{rec}}$ such that
        $\mathrm{wt}_{\min}(E_{\mathrm{rec}} \cdot E) \leq 2|u|$.
    \end{cor}

    The intuition behind this result is given by performing essentially the same procedure that we performed on the classical parity check matrix in section~\ref{sec:decoding} to put it in its canonical form. It can be shown that given a set of stabilizer generators $G = \{g_i\}$ on $n$ qubits, it is possible to find a set of stabilizer generators equivalent up to qubit relabeling and local Cliffords such that for every $i$, qubit $i$ is only in the support of $g_i$. This effectively diagonalizes the stabilizer generators in the same manner as the row operations and column permutations do in the classical case. Upon following this procedure, it is no longer possible for a measurement error to prompt the application of a high-weight, non-fault-tolerant, correction under minimum weight decoding. This is formalized using Campbell's definition of ``soundness''.
    
    In this section we show that this result applies to our construction, even though we do not measure a complete set of generators for the original code.
    
    Our construction applied to a code $Q$ is sufficient to ensure any codestate $q \in Q$ remains in $Q$ after errors and error correction, which justifies conceptualizing of our sequence of measurements as a fault-tolerant measurement schedule for $Q$. However, if we ignore the underlying structure and just consider the code defined by the stabilizers we actually measure we get a different, but related, code. 
    
    Although this codespace $Q'$ contains $Q$, in general it is larger than $Q$. The additional degrees of freedom are the stabilizers of $Q'$ along with their associated destabilizers or pure errors; the destabilizer $d_i$ associated with the stabilizer $g_i$ is just the unique-up-to-stabilizers operator such that $[g_j, d_i] = \delta_{i,j}$. In this context, $Q'$ is a code defined by a stabilizer group $G'$, which is a subgroup of the stabilizer group defining $Q$, and a gauge group, exactly half of which are stabilizers of $Q$. Fixing the gauge of $Q'$ produces $Q$. 
    
    Considering $Q'$ as a code in its own right shows that Campbell's construction can be applied to achieve single-shot fault tolerance. In this construction, a different set of stabilizer generators of $Q'$ are measured, with the property that for any syndrome of weight one there exists a unique data error of weight one which produces that syndrome. By Campbell's soundness property this is sufficient to ensure single-shot fault tolerance. 
    
    Applying Campbell's result to our construction therefore lifts our result about single-round detectability to a result about single-round fault-tolerance. 
    The errors which are not caught by this schedule of measurements are still just the logical or gauge operators of $Q'$, which of course are either harmless stabilizers or guaranteed to be high-weight by Theorem~\ref{lem:detectability}.
    \begin{cor}
        Applying Campbell's construction to the code defined by our construction yields a sequence of $O(d \log (n - k))$ measurements sufficient for fault-tolerant error correction of an $[[n, k, d]]$ LDPC code, or $O(d \log(d) \log(n - k))$ measurements sufficient for fault-tolerant error correction of an $[[n, k, d]]$ concatenated code.
    \end{cor}
    This justifies our previous statement that this work constructively demonstrates the length $O(d \log(n - k))$ measurement schedule which was previously known 
     only to exist.
    
     We note that we do not expect this construction to achieve a threshold for most codes -- Brown et. al \cite{brown2024single} have shown numerical evidence that a threshold does not exist even just applying Campbell's construction to the surface code, and we do not expect that adding our measurement reduction procedure to substantially improve this picture. We discuss the impact of high-weight stabilizers briefly in appendix~\ref{app:stabilizer_weight}.
     
     \section{Data-Syndrome Codes}\label{sec:data_syn}
     Our work bears some resemblance to the concept of \emph{data-syndrome codes} \cite{ashikmhin2014robus,brown2024datasyndromebch, fujiwara2014datasyndrome}, by which we can make our measurement protocol resistant to measurement errors using classical codes. 
     The resemblance is misleading, however, in that our construction aims to do \emph{the opposite} of data-syndrome codes.
     In the data-syndrome code picture, we notice that one can describe measuring a stabilizer over $d$ rounds as encoding the measurement result in a repetition code with distance $d$; concretely, this corresponds to multiplying a bit representing the stabilizer by the generator matrix for the repetition code. Since we perform the same procedure for each stabilizer generator, the generator matrix describing repeated rounds of measurement is given by $G_d^T \otimes I_{\ell}$, where $G_d$ is the generator matrix for the distance $d$ repetition code, and $\ell = n - k$ as before. Then the sequence of measurements performed is given by $(G_d^T \otimes I_\ell) H$ where $H$ is the stabilizer matrix for the quantum code in question and the multiplication is over $GF(4)$. We interpret each column $c_j$ of the measurement of the stabilizer $\prod_{i\in I} g_i$ where $I$ is the set of all $i$ such that $c_{ij} = 1$. 
      The fact that we use a distance $d$ repetition code means that given $s$ data errors and $p$ measurement errors, we can correctly identify the measurement errors as long as $s + p \leq \lfloor d/2\rfloor$ by first identifying measurement errors using $G_d$ then identifying the data error using the corrected syndrome.
    
    This framing naturally suggests replacing $G_d \otimes I_\ell$ with the generator matrix for a different $[n, \ell, d]$ classical code. Generically this can save many measurements since we only need $d\times \ell < n$ to improve upon the strategy of repeated rounds of measurement. It is worth clarifying exactly how our work differs from this strategy. In our work, we use the parity check matrix of an $[\ell, k, cd]$ classical code, assuming that such a code exists, where $k$ is as large as possible to reduce the number of measurements from $\ell$ to $\ell - k$. We then either repeat $d$ times, or follow the strategy given in section~\ref{sec:single_shot} to make this sequence of measurements fault tolerant. Choosing a classical code with few parity checks directly translates to making few measurements. Instead of encoding the syndrome bits, we compress them. When using a data-syndrome code the number of measurements made, $n$, can never be smaller than the original number of stabilizer generators by the singleton bound $d + \ell \leq n + 1$. Instead of compressing the syndrome bits and then repeating for redundancy, data-syndrome codes are a strategy to make the syndrome bits redundant.
    
    Therefore, these two strategies solve complementary parts of the error correction picture. Following Brown et al. \cite{brown2024datasyndromebch} we can apply the generator matrix for the $[n, k = O(d \log \ell), d]$ BCH code to our construction, where $n = O(k + d \log k) = O(d \log \ell + d \log (d \log \ell)) = O(d \log \ell)$. Concretely, we take the generator matrix for the given BCH code $G_{BCH}$ and simply consider $G_{BCH} P H_x$ and $G_{BCH} P H_z$ (or consider multiplication over $GF(4)$ for the non-CSS case) to yield a measurement schedule of length $O(d \log \ell)$.
    
    This measurement schedule corrects any data error of weight $s$ even in the presence of $p$ measurement errors as long as $s + p < \lfloor d/2 \rfloor$. It is extremely important to realize, however, that this construction is not necessarily fault tolerant, weakly or otherwise. In previous work on data-syndrome codes \cite{ashikmhin2014robus,brown2024datasyndromebch, fujiwara2014datasyndrome}, the authors only prove that the resulting measurement scheme is robust to measurement errors, and do not consider a noise model which includes mid-circuit errors. Therefore despite achieving the same or better scaling as other methods we have proposed in this work, it is not directly comparable since it solves only a portion of the same problem.
    \section{Conclusion} \label{sec:conclusion}
    In this work we have shown, although stabilizer codes
    are often conflated with their stabilizer generators, that if one is given a stabilizer codestate with errors it is possible to identify the errors by measuring a particular set of stabilizers that do not span the codespace. This can be understood by observing that many codespaces are actually larger codespaces up to fixing a gauge. Our proof is constructive, in that a set of measurements sufficient to identify errors are defined in terms of the original stabilizer generators combined according to a classical code. This procedure allows us to relate the ability of our proposed set of measurements to detect errors to the abilities of the classical and quantum codes.
    
    The procedure we give for LDPC codes or concatenated codes is extremely modular, and can be modified to optimize for a variety of desirable characteristics, such as check weight or locality. We illustrate this in section~\ref{sec:alternative} where we give brief examples to show that a reduction in the number of measurements required for the surface code is possible by measuring low-weight stabilizers, and that by choosing subsets of stabilizer generators of an LDPC code carefully even more measurements can be eliminated than in our initial construction. We leave as an open question how to optimize simultaneously for stabilizer weight and number of measurements by careful choice of 1) stabilizer generators, 2) the ordering of stabilizer generators, and 3) the subsets of stabilizer generators chosen for compression.
    
    We also leave the question of efficient decoding mostly open. In our numerical experiments in section~\ref{sec:surface_code} we used an inefficient brute force decoder. In theory, one should expect that the large amount of structure inherent in our construction should admit efficient decoders which take advantage of this structure. In section~\ref{sec:decoding} we propose one efficient two-step decoder, which preserves the distance of the construction under a noise model with measurement errors. However, the qualitative behavior is very different from our brute force decoder, and more work needs to be done to optimize the classical decoder to work well with the quantum decoder.
    
    Finally, we note our method to make our sequence of measurements single-shot fault tolerant is not the only avenue for fault tolerance other than repetition. Applying data-syndrome codes to our construction yields a sequence of measurements which, although not fault tolerant, corrects data errors while being resistant to measurement errors. Extending this construction to account for mid-circuit noise could also yield a fault-tolerant measurement schedule of $O(d \log \ell)$ measurements.
    \section{acknowledgments}
    This work is supported by NSF CAREER award No. CCF-2237356. We would like to thank the UNM Center for Advanced Research Computing, supported in part by the National Science Foundation, for providing the high performance computing resources used in this work. 
    
    \appendices
    \section{BCH Codes}\label{appendix:BCH}
    The parity check matrix of the BCH code can be defined by considering primitive elements of a certain finite field. A primitive element of $GF(2^m)$ is an element $\alpha \in GF(2^m)$ such that every nonzero element of $GF(2^m)$ can be written as $\alpha^i$ for some $i$. The relevant finite field is obtained by taking $m$ to be $\lceil{\log_2 w}\rceil$; the binary parity check matrix then has $mt$ rows, which we will show by constructing the matrix.
    
    The parity check matrix $H'$ can be defined as
    \[
        H' := \begin{pmatrix}
        1 & \alpha      & \alpha^2        & \ldots & \alpha^{w - 1}\\
        1 & \alpha^2    & (\alpha^2)^2    & \ldots & (\alpha^3)^{w - 1}\\
          &             &                 &\vdots\\
        1 & \alpha^{2t} & (\alpha^{2t})^2 & \ldots & (\alpha^{2t})^{w - 1}
        \end{pmatrix}.
    \]
    Every other row of $H$ is redundant since $\sum (\alpha^i)^2 = 0$ iff $\sum \alpha^i = 0$ since $GF(2^m)$ has characteristic two. We can then omit every second row so that $H_{ij} = (\alpha^{2j + 1})^{i + 1}$, for $i,j$ starting at zero.
    
    Taking the binary reprentation of (elements of) $GF(2^m)$ produces a parity check matrix with $mt$ rows. By our choice of $m$ this produces $t\lceil{\log_2 w}\rceil$ rows (or parity checks).
    
    We want to show that $H$ defines a code with distance $d = 2t + 1$. This is equivalent to any set of $2t$ columns of $H$ being linearly independent since this means that any set of $2t$ errors is detectable, hence meaning that any two errors each with weight at most $t$ are distinguishable. Since $H'$ defines the same code as $H$, we can work with it instead for convenience. Suppose for the sake of contradiction that $H'v = 0$ for some $v$ such that $|{v}| \leq 2t$. Then 
    \[\begin{pmatrix}
        \alpha^{j_1} & \alpha^{j_2} & \ldots & \alpha^{j_{|{v}|}}\\
        (\alpha^{j_1})^2 & (\alpha^{j_2})^2 & \ldots & (\alpha^{j_{|{v}|}})^2\\
                     &             & \ddots &                   \\
        (\alpha^{j_1})^{2t} & (\alpha^{j_2})^{2t} & \ldots & (\alpha^{j_{|{v}|}})^{2t}\\
    \end{pmatrix} \begin{pmatrix}
    1 \\ 1 \\ \vdots \\ 1
    \end{pmatrix} = 0\] for all $j_i$ such that $v_{j_i} = 1$.
    Since $|v| \leq 2t$ we can just consider the first $|v|$ equations in this system to produce a square matrix. Factoring out a power of $\alpha^{j_i}$ from the $j_i$ column shows that this equation reduces to the determinant of a Vandermonde matrix being equal to zero, which is impossible. Therefore $H'$, and hence $H$, defines a code with distance at least $2t + 1$.
    \section{High-weight Stabilizers}\label{app:stabilizer_weight}
    It is often acknowledged in a folk-loric sense that if stabilizer weights are too high no threshold will exist for the code they belong to, but we are unaware of a precise statement and proof of this fact. We would like to formalize this intuition with the following lemma.
    \begin{lemma}
        If each two-qubit gate produces a two-qubit depolarizing error with parameter $p$, then measuring an operator of weight $w$ with Shor-style syndrome extraction requires $O(2^w)$ repetitions to ensure the measurement result is correct with probability $>0.99$.
    \end{lemma}
    \begin{proof}
        Each of the $w$ two-qubit gates introduces an error detected by the measurement with probability $p' = p/16$. The probability that, after all $w$ of these gates, there is no error on the ancilla state is then $\frac{1 + (1 - 2p)^w}{2}$. This can be obtained by examining the recurrence relation $P_x = p + (1 - 2p)P_{x - 1}$ where $x$ is the number of successes.

        Therefore, by Hoeffding's inequality, the square root of the number of repetitions must dominate $\frac{2}{(1 - 2p)^w}$ in order to ensure the bias can be discerned with high probability.
    \end{proof}

    The point of this lemma is that if we wish to obtain a threshold, it is intuitive that the measurement results must be arbitrarily reliable (although in some sense this assumption is violated by the fact that subsystem codes have random measurement results with only deterministic products). If we ask for this, and the weight of the stabilizers grows with the distance, the number of repetitions required grows exponentially with the distance as well. This means that the time spent idling is exponential with the distance, so the effective physical error rate grows with the same scaling as the distance. However, this limited analysis does not explain why concatenated codes, which have stabilizers of weight $O(d)$, can produce a threshold. 
    
    It is also interesting to note that this analysis makes explicit the fact that logarithmically growing stabilizer weight permits polynomially long measurement sequences to yield reliable measurement results. To be explicit, the above lemma shows that $O(2^w)$ measurements of a weight $w$ operator is necessary for the probability that the majority of the measurements is correct to be greater than $0.99$. If $w = c\log(d)$ then the number of measurements necessary is $O(2^{c \log d}) = \text{poly}(d)$.

    It is also interesting to note that this only applies to Shor-style syndrome extraction. In Knill- or Steane- style syndrome extraction, any weight stabilizer has the same measurement error rate, since their measurement results are computed in a classical, error-free, manner from single qubit measurements of an ancilla codestate.

    \bibliographystyle{apsrev4-2}
    \bibliography{refs.bib}
    \end{document}